\documentclass[11pt,a4paper]{article}
\usepackage[T1]{fontenc}
\usepackage[utf8]{inputenc}
\usepackage{lmodern,microtype}
\usepackage[margin=25mm]{geometry}
\usepackage{amsmath,amssymb,amsthm,mathtools,braket}
\usepackage{booktabs,graphicx,xcolor}% Figure support and float layout.
\usepackage{tikz}
\usetikzlibrary{arrows.meta,calc,positioning,decorations.pathreplacing}
\definecolor{lcBlue}{HTML}{2463A0}
\definecolor{lcOrange}{HTML}{C66A23}
\definecolor{lcTeal}{HTML}{158477}
\definecolor{lcPurple}{HTML}{82559A}
\definecolor{lcInk}{HTML}{263444}
\definecolor{lcGray}{HTML}{A6B0BA}
\definecolor{lcLight}{HTML}{EDF1F5}
\tikzset{lc base/.style={font=\sffamily\small,text=lcInk,line cap=round,line join=round},
 lc title/.style={font=\sffamily\bfseries\small,anchor=west},
 lc note/.style={font=\sffamily\footnotesize,align=center},
 lc box/.style={draw=lcGray,rounded corners=2pt,fill=white,align=center,inner sep=5pt},
 lc arrow/.style={-{Latex[length=2mm]},draw=lcInk,line width=.7pt},
 lc dot/.style={circle,draw=lcGray,fill=white,inner sep=0pt,minimum size=3.4pt},
 lc aux/.style={rectangle,draw=lcInk,fill=white,inner sep=0pt,minimum size=5pt}}

\usepackage[section]{placeins}
\usepackage{authblk}
\usepackage[font=small,labelfont=bf]{caption}
\usepackage[numbers,sort&compress]{natbib}
\usepackage[colorlinks=true,allcolors=blue]{hyperref}
\usepackage[capitalise,noabbrev]{cleveref}
\newtheorem{theorem}{Theorem}
\newtheorem{proposition}[theorem]{Proposition}
\newtheorem{lemma}[theorem]{Lemma}
\theoremstyle{remark}\newtheorem{remark}[theorem]{Remark}
\newcommand{\F}{\mathbb F}
\newcommand{\GL}{\operatorname{GL}}
\newcommand{\Sp}{\operatorname{Sp}}
\newcommand{\wt}{\operatorname{wt}}
\newcommand{\supp}{\operatorname{supp}}
\newcommand{\im}{\operatorname{im}}

\newcommand{\Span}{\operatorname{span}}
\newcommand{\Sym}{\operatorname{Sym}}
\newcommand{\Xl}{\overline X}
\newcommand{\Zl}{\overline Z}

\newcommand{\src}{\mathrm{src}}

\theoremstyle{definition}\newtheorem{example}[theorem]{Example}
\crefname{example}{Example}{Examples}

\title{Fast computational primitives for La-cross qLDPC codes}

\author[1]{Andrea Di Fini$^{\dagger}$}
\author[1,2]{Tom Hartweg$^{\dagger}$}
\author[1]{Laura Pecorari}
\author[1,2,3]{Guido Pupillo}
\author[2]{Samuel Crew$^{*}$}

\affil[1]{University of Strasbourg and CNRS, CESQ and ISIS (UMR 7006), aQCess, 67000 Strasbourg, France}

\affil[2]{QPerfect SAS, European Center for Quantum Sciences, 23 rue du Loess, 67200 Strasbourg, France}

\affil[3]{Institut Universitaire de France (IUF), 75000 Paris, France}

\date{\today}

\begin{document}
\maketitle

\begingroup
\renewcommand{\thefootnote}{\fnsymbol{footnote}}
\footnotetext[2]{These authors contributed equally.}
\footnotetext[1]{Corresponding author: \texttt{samuel.crew@qperfect.io}}
\endgroup

\begin{abstract}
Quantum LDPC codes promise to reduce the overhead of fault-tolerant quantum computation, but their practical use requires efficient logical gates, measurements, and access to non-Clifford resources. Here, we develop a complete set of computational primitives for La-cross hypergraph product codes. For a subclass of these codes, we show that seed automorphisms together with fold-transversal gates generate the full logical Clifford group up to Pauli corrections. We further prove that every rank one logical operator admits $D$ pairwise disjoint physical representatives, enabling parallel gauging and a space--time tradeoff for logical measurements. Finally, we construct an addressable surface code interface for injecting externally prepared resource states into selected logical qubits. Together, these results establish La-cross codes as a simple, low-connectivity generalisation of surface code computation, retaining analogous computational primitives while providing a complete Clifford+$T$ instruction set across multiple encoded qubits.
\end{abstract}

\section{Introduction}

Progress in quantum low density parity check (qLDPC) codes has recently enjoyed remarkable acceleration \cite{breuckmann2021quantum,bravyi2024high}. Lifted product constructions established asymptotically good quantum codes with constant rate and linear distance \cite{panteleev2022asymptotically}, while subsequent work has developed increasingly practical finite size families and methods for computation with them \cite{bhardwaj2026high}. Code surgery provides a generally applicable route to logical measurements \cite{horsman2012surface,williamson2026low,cowtan2026parallel,di2026error}, and recent fast and lifted surgery constructions substantially reduce its overhead \cite{cohen2022low,baspin2025fast,berent2026lifted}. These developments will eventually lead to large scale fault tolerant architectures combining the spatial efficiency of qLDPC memories with efficient logical computation. Here we study a near-term regime with small finite size codes whose structure natively supports fast logical operations directly within the encoded block. This is particularly relevant to neutral-atom architectures, where qubit movement, long-range interactions and parallel operations provide flexible connectivity, but atom rearrangement and repeated syndrome extraction can introduce significant time overhead \cite{pecorari2025high,xu2024constant}.

The use of code automorphisms to implement fault-tolerant logical Clifford operations has been studied in \cite{grassl2013leveraging,sayginel2025fault}. The native automorphisms and fold-transversal gates that we study in the present work can avoid the distance dependent measurement time of conventional surgery and, for Clifford operations, the auxiliary logical blocks required by teleportation based gate gadgets \cite{xu2025fast}. Our aim is to identify compact codes with a fast and structurally simple native instruction set for implementation on present day neutral atom hardware. By `fast' we mean that the Clifford operations are realised natively through code automorphisms and fold-transversal gates, while parallel gauging reduces logical Pauli measurements to a constant number (independent of the code distance) of syndrome extraction rounds. 

La-cross codes \cite{pecorari2025high} are a particularly simple example of hyper graph product (HGP) qLDPC codes. As we elucidate further in this work, they may be viewed as a minimal generalisation of a surface code \cite{hong2024long}. A surface code patch is obtained from a hyper graph product of two repetition codes, while a La-cross code replaces each repetition code by a (suitably truncated) classical recurrence code. The stabilisers retain the surface code cross shape with two additional `legs' of length $k$, resulting in a grid of $k\times k$ logical qubits, instead of the single logical qubit of a surface code \cite{pecorari2025addressable}. The familiar row--column structure of logical operators is also retained. In this work, we show that this analogy extends from memory to computation. The HGP structure supplies natural Hadamard and phase folds, while the automorphism group is sufficiently large to enable Clifford computation. Efficient logical Clifford computation has also been obtained in specially constructed qLDPC code families using transversal operations \cite{malcolm2026computing}.

Previous work obtained addressable Clifford operations by coupling a La-cross block to an auxiliary code gadget \cite{pecorari2025addressable}. Here we identify finite La-cross codes whose intrinsic automorphisms and fold-transversal gates generate the full logical Clifford group, and construct an addressable surface code interface for injecting externally prepared $\ket T$ states. This therefore gives a complete Clifford+$T$ instruction set for these blocks. 

\paragraph{Contributions.}
Our main results are:
\begin{itemize}
    \item \textbf{Intrinsic Clifford operations.}
    We determine the logical action of classical seed automorphisms of La-cross codes and show that, for primitive complete period La-cross codes, these automorphisms together with Hadamard and $S$ folds generate the full logical Clifford group (\cref{thm:clifford}).

    \item \textbf{Parallel surgery.}
    We prove that every rank one logical Pauli in the primitive complete period family has $D$ pairwise disjoint equivalent representatives (\cref{prop:disjoint}). Gauging these representatives in parallel gives a space-time tradeoff for logical measurement. The construction gives a two round measurement protocol using $O(D^2\log D)$ ancillary qubits (\cref{sec:measure}).

    \item \textbf{Surface code interface.}
    We construct a homomorphic \textsf{CNOT} interface from a distance $D$ surface code patch to an addressed La-cross logical degree of freedom and prove its fault distance and implementation depth. We thereby demonstrate a method to transfer externally prepared surface code resource states into the La-cross block (\cref{prop:depth,prop:transfer-fd}).
\end{itemize}
Together, these primitives provide a complete Clifford+$T$ instruction set for La-cross codes, combining native Clifford operations, parallel logical measurements and addressable resource state injection. They show that compact La-cross code memories can support fast, universal logical computation.

\paragraph{Outline.}
The paper is organised as follows. In \Cref{sec:codes} we review La-cross codes, their logical operators and discuss the primitive complete period family. In \Cref{sec:clifford} we discuss La-cross automorphisms and fold-transversal Clifford gates. In \Cref{sec:measure} we discuss a parallel surgery method for rank one logical measurements, and finally in \cref{sec:interface} we construct the addressable surface code interface.

\section{La-cross codes and logical representatives}\label{sec:codes}

In this section we recall the structure of the open boundary La-cross codes used throughout the paper. We first describe the general hypergraph product construction and its logical operators, and then specialise to a family of La-cross codes where the logical representatives have additional structure that will be useful later for computation.

La-cross codes \cite{pecorari2025high} are hypergraph product (HGP) codes \cite{tillich2013quantum} obtained from a classical seed code with itself. Let $n>k\ge2$, and write $m=n-k$. The classical seed code is defined by the parity check matrix\footnote{Here $H_{a,*}$ denotes the $a$th row of $H$, and $e_j\in\F_2^n$ denotes the standard basis vector with a single non-zero entry in position $j$.}
\begin{equation}\label{eq:seed}
    H=H_{k,n}\in\F_2^{m\times n}, \qquad H_{a,*}=e_a+e_{a+1}+e_{a+k}, \quad 0\le a<m .
\end{equation}
Thus each check acts on three coordinates with offsets $(0,1,k)$, with open boundary conditions so that no check wraps around the ends of the seed. The first $m$ columns of $H$ contain an upper-triangular $m\times m$ block with unit diagonal, so $H$ has full row rank $m$. Consequently $C_{k,n}=\ker H$ is an $[n,k]$ classical code. We write $d(C_{k,n})$ for the classical distance. A codeword $c\in C_{k,n}$ satisfies
\begin{equation}
    c_a+c_{a+1}+c_{a+k}=0, \qquad 0\le a<m.
\end{equation}
Equivalently, the first $k$ coordinates may be chosen freely, and the remaining $m$ coordinates are determined recursively by 
\begin{equation}\label{eq:seed-recurrence}
    c_{a+k}=c_a+c_{a+1}.
\end{equation}

We write the generator matrix in the form
\begin{equation}
G=\begin{pmatrix}I_k&A\end{pmatrix}\in\F_2^{k\times n},
\end{equation}
with each row completed from its first $k$ entries using the recurrence relation \eqref{eq:seed-recurrence}. This fixes a basis for the $k$-dimensional seed code and will also be useful below when describing seed automorphisms through their action on the columns of $G$.

We introduce the following notation: for $u\in\F_2^k$, $c_u=G^Tu\in C_{k,n}$ is the seed codeword whose first $k$ entries are $u$, and we write the columns as $g_j=Ge_j\in\F_2^k$ so that $(c_u)_j=g_j\cdot u=g_j^Tu$. Thus $g_j$ records how the $j$th coordinate of a codeword depends linearly on the initial condition $u$. The recurrence of equation \eqref{eq:seed-recurrence} is then equivalently
\begin{equation}
    g_{a+k}=g_a+g_{a+1}.
    \label{eq:column-recurrence}
\end{equation}
It is convenient to encode this recurrence by the polynomial
\begin{equation}
    h_k(x)=1+x+x^k .
    \label{eq:row-polynomial}
\end{equation}
Identifying $\F_2^k$ with $\F_2[x]/(h_k)$ in the basis $1,x,\ldots,x^{k-1}$, we have
\begin{equation}
x^k=1+x \pmod{h_k},
\end{equation}
and hence
\begin{equation}
x^{a+k}=x^a+x^{a+1}\pmod{h_k}.
\end{equation}
This is precisely the column recurrence \eqref{eq:column-recurrence}, with the initial columns $g_0,\ldots,g_{k-1}$ identified with $1,x,\ldots,x^{k-1}$. Under this identification,
\begin{equation}\label{eq:polmaps}
    g_j=x^j\bmod h_k .
\end{equation}
Following \cite{pecorari2025high,pecorari2025addressable}, we use $h_k$ as the row, or recurrence, polynomial.\footnote{With the standard cyclic code convention, the corresponding check polynomial (before the open-boundary truncation) is the reciprocal $h_k^*(x)=x^k h_k(x^{-1})=1+x^{k-1}+x^k$.}

The La-cross quantum code is then defined as the HGP of $H$ with itself and consequently has parity check matrices
\begin{equation}
    H_X=[H\otimes I_n\mid I_m\otimes H^T], \qquad H_Z=[I_n\otimes H\mid H^T\otimes I_m],
    \label{eq:checks}
\end{equation}
and corresponding parameters
\begin{equation}
    [[n^2+m^2,k^2,d(C_{k,n})]].
\end{equation}
Its physical qubits form two square arrays. An $n\times n$ \emph{principal sector}, indexed by pairs of seed coordinates, and an $m\times m$ \emph{check sector}, indexed by pairs of seed checks. The stabilisers have weight at most six. In the interleaved layout proposed in \cite{pecorari2025high}, a bulk weight six check can be viewed as a surface-code cross involving four nearby qubits, together with two longer range legs whose extent is set by $k$ and is independent of $n$.

We refer to the qubits with fixed first coordinate in the principal sector a \emph{principal row}, and those with fixed second coordinate a \emph{principal column}. As in the surface code, the $k^2$ logical qubits admit canonical $X$ representatives along principal rows and canonical $Z$ representatives along principal columns; see \cref{fig:logical-representatives}.  

It is convenient to arrange the $k^2$ logical qubits as a $k\times k$ grid and label logical Paulis by matrices $M,N\in\F_2^{k\times k}$.\footnote{More precisely, the principal sector $X$ logicals are $(\F_2^n/\operatorname{im}H^T)\otimes\ker H$, with the two factors reversed for $Z$ logicals. $H$ has full row rank so both factors may be identified with $\F_2^k$. We use the identification $q\otimes u\leftrightarrow qu^T$.} The matrix $E_{ij}=e_ie_j^T$ labels logical qubit $(i,j)$ with $0\leq i,j<k$, so that $\Xl_{ij}=\Xl(E_{ij})$ and $\Zl_{ij}=\Zl(E_{ij})$, with the commutation relation
\begin{equation}
\Xl(M)\Zl(N) = (-1)^{\operatorname{Tr}(MN^T)} \Zl(N)\Xl(M).
\end{equation}

We represent a physical $X$-type Pauli operator by a pair of binary matrices $(A,B)$, with $A\in\F_2^{n\times n}$ and $B\in\F_2^{m\times m}$, specifying its support on the $n\times n$ principal sector and the $m\times m$ check sector, respectively, and write $X(A,B)$ for the corresponding Pauli operator.\footnote{Explicitly, $X(A,B)=\prod_{i,j}X_{(i,j)}^{A_{ij}} \prod_{a,b}X_{[a,b]}^{B_{ab}}$, where $(i,j)$ labels principal sector qubits and $[a,b]$ labels check sector qubits.}
Similarly, $Z(A,B)$ denotes the analogous $Z$-type operator.

Writing $\epsilon_i\in\F_2^n$ for the $i$th physical coordinate vector, canonical representatives may be expressed as
\begin{equation}
\Xl_{ij}\sim X(\epsilon_i c_{e_j}^T,0), \qquad \Zl_{ij}\sim Z(c_{e_i}\epsilon_j^T,0).
\label{eq:canonical-reps}
\end{equation}
Thus $\Xl_{ij}$ lies on principal row $i$ with pattern $c_{e_j}$, while $\Zl_{ij}$ lies on principal column $j$ with pattern $c_{e_i}$. 

Since the first $k$ entries of $c_{e_j}$ are $e_j$, these representatives satisfy 
\begin{equation}
\left\langle \epsilon_i c_{e_j}^T,\, c_{e_p}\epsilon_q^T \right\rangle = (\epsilon_i^T c_{e_p}) (c_{e_j}^T\epsilon_q) = \delta_{ip}\delta_{jq},
\end{equation}
where $\langle A,B\rangle=\operatorname{Tr}(AB^T)$ is the binary inner product on support matrices. Hence
\begin{equation}
X(\epsilon_i c_{e_j}^T,0) Z(c_{e_p}\epsilon_q^T,0) = (-1)^{\delta_{ip}\delta_{jq}} Z(c_{e_p}\epsilon_q^T,0) X(\epsilon_i c_{e_j}^T,0).
\end{equation}
These operators therefore realise the canonical Pauli algebra of the $k^2$ logical qubits. The top-left $k\times k$ part of the principal sector provides a natural coordinate grid for the $k^2$ logical qubits, as illustrated in \cref{fig:logical-representatives}.

Every rank-one matrix $M\in\F_2^{k\times k}$ may be written as $M=qu^T$ for nonzero $q,u\in\F_2^k$. Since $G$ has full row rank, for any $q$ we may choose $v\in\F_2^n$ satisfying $Gv=q$. Then
\begin{equation}
X(v c_u^T,0),\qquad c_u=G^Tu,
\label{eq:equivalent-log-pauli}
\end{equation}
represents the logical operator $\Xl(qu^T)$.

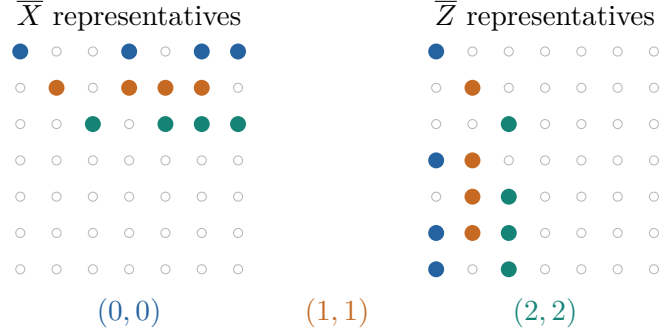
\begin{figure}[htbp]
\centering
\begin{tikzpicture}[
    x=0.48cm, y=-0.48cm,
    qubit/.style={
        circle, draw=gray!55, fill=white,
        inner sep=0pt, minimum size=3.4pt
    }
]
\definecolor{logicalBlue}{HTML}{2463A0}
\definecolor{logicalOrange}{HTML}{C66A23}
\definecolor{logicalGreen}{HTML}{158477}

% X representatives.
\begin{scope}
    \node at (3,-1) {$\overline X$ representatives};
    \foreach \i in {0,...,6}{
        \foreach \j in {0,...,6}{
            \node[qubit] at (\j,\i) {};
        }
    }
    \foreach \j in {0,3,5,6}
        \fill[logicalBlue] (\j,0) circle (3pt);
    \foreach \j in {1,3,4,5}
        \fill[logicalOrange] (\j,1) circle (3pt);
    \foreach \j in {2,4,5,6}
        \fill[logicalGreen] (\j,2) circle (3pt);
\end{scope}

% Z representatives.
\begin{scope}[xshift=5.5cm]
    \node at (3,-1) {$\overline Z$ representatives};
    \foreach \i in {0,...,6}{
        \foreach \j in {0,...,6}{
            \node[qubit] at (\j,\i) {};
        }
    }
    \foreach \i in {0,3,5,6}
        \fill[logicalBlue] (0,\i) circle (3pt);
    \foreach \i in {1,3,4,5}
        \fill[logicalOrange] (1,\i) circle (3pt);
    \foreach \i in {2,4,5,6}
        \fill[logicalGreen] (2,\i) circle (3pt);
\end{scope}

% Shared logical-qubit labels.
\node[text=logicalBlue]   at (3,7.2)    {$(0,0)$};
\node[text=logicalOrange] at (8.73,7.2) {$(1,1)$};
\node[text=logicalGreen]  at (14.46,7.2){$(2,2)$};
\end{tikzpicture}

\caption{Canonical logical representatives in the $[[65,9,4]]$
La-cross code. Colours distinguish three logical qubits:
$(0,0)$, $(1,1)$ and $(2,2)$. Their $X$ representatives occupy principal rows (left), while their $Z$ representatives occupy principal columns (right). Only the $7\times7$ principal sector is shown.}
\label{fig:logical-representatives}
\end{figure}

The principal sector also controls the distance relevant to the fault tolerance arguments below.

\begin{lemma}[Principal support]\label{lem:principal}
Let $H$ have full row rank and let $C=\ker H$ have distance $d(C)$. For the square hypergraph product $\operatorname{HGP}(H,H)$, every nontrivial logical $X$ occupies at least $d(C)$ columns of the principal sector, while every nontrivial logical $Z$ occupies at least $d(C)$ principal rows. In particular, every nontrivial logical Pauli has principal sector weight at least $d(C)$.
\end{lemma}

\begin{proof}
This follows directly from the logical row--column weight bound of
\cite{quintavalle2022reshape}
\end{proof}

\subsection{Primitive complete periods family}
A particularly interesting family arises when $h_k(x)=1+x+x^k$ is a primitive polynomial of degree $k$ over $\mathbb{F}_2$.\footnote{A degree-$k$ polynomial $h_k(x)\in\F_2[x]$ is primitive if it is irreducible and if the residue class of $x$ generates the multiplicative group of the field $\F_2[x]/(h_k)\cong \F_{2^k}$. Since $\F_{2^k}^\times$ has order $2^k-1$, the powers $1,x,x^2,\ldots,x^{2^k-2}$ modulo $h_k$ run through all nonzero elements of $\F_{2^k}$ exactly once before repeating.
}
We define
\begin{equation}
 L=2^k-1,\qquad n=sL,\qquad D=s2^{k-1}.
 \label{eq:primitive}
\end{equation}
Since $g_j=x^j\bmod h_k$, the column sequence of $G$ is $L$-periodic and, within each period, runs once through every nonzero vector of $\F_2^k$.

Hence, for any nonzero $u\in\F_2^k$, the codeword components $(c_u)_j=g_j^Tu$ are nonzero on exactly $2^{k-1}$ columns per period. Indeed, $v\mapsto u^Tv$ is a nonzero linear functional on $\F_2^k$, so its kernel is a $(k-1)$-dimensional subspace containing $2^{k-1}$ vectors. Therefore it takes the value $1$ on the remaining $2^{k-1}$ vectors of $\F_2^k$. Since one primitive period of the sequence $g_j$ contains every nonzero vector of $\F_2^k$ exactly once, and $0$ lies in the kernel, exactly $2^{k-1}$ columns in each period satisfy $g_j^Tu=1$. Therefore, every nonzero seed codeword has weight $D=s2^{k-1}$ and the seed distance is $d(C_{k,n})=D$.
 
Since $H$ has full row rank, the corresponding square HGP has parameters
\begin{equation}
[[N,K,D]] = [[n^2+(n-k)^2,k^2,s2^{k-1}]].
\end{equation}
The first two examples are $[[65,9,4]]$ and $[[346,16,8]]$ and we discuss the first of these as a representative example in detail in the following.

\begin{proposition}[Disjoint representatives]\label{prop:disjoint}
Every rank one logical $X$ or $Z$ in the primitive complete period family has $D$ pairwise disjoint equivalent representatives, each of weight either $D$ or $2D$.
\end{proposition}

\begin{proof}
See appendix \ref{app:disjoint-proof}.
\end{proof}

\begin{example}
The primitive complete period condition is sufficient but not necessary for this property of $D$ disjoint representatives. With direct enumeration we find further La-cross examples outside the primitive case with this property. For example, the reducible $k=5$ seeds at $n=15$ and $n=18$ give La-cross codes
\begin{equation}
[[325,25,5]] \qquad\text{and}\qquad [[493,25,7]],
\end{equation}
respectively, and in both cases we find that every rank one logical $X$ or $Z$ has $D$ pairwise disjoint representatives.\footnote{Numbering the seed coordinates $0,\ldots,n-1$, for $q=e_0$ one may choose the disjoint subsets
\begin{equation}
\{0\},\ \{1,5\},\ \{2,10\},\ \{11,13\},\ \{6,9,12\}
\end{equation}
for $n=15$, and
\begin{equation}
\{0\},\ \{1,5\},\ \{2,10\},\ \{11,13\},\ \{16,17\},
\ \{6,9,12\},\ \{8,14,15\}
\end{equation}
for $n=18$. In each subset the corresponding generator columns sum to $e_0$. Hence their indicator vectors $v$ satisfy $Gv=e_0$ and give disjoint equivalent representatives $X(v c_u^T,0)$ of $\overline X(e_0u^T)$ for any nonzero $u$. The other $q$ and the corresponding $Z$ representatives are obtained similarly.}
Thus the property extends also beyond primitive complete periods.
\end{example}

We note, however, that the $D$ disjoint representative property does not hold for every La-cross code. For example, the $(k,n)=(5,20)$ code has distance nine, but the rank one logical $\overline X_{0,4}$ has exactly eight pairwise disjoint equivalent representatives.\footnote{Here $c_{e_4}$ has support ${4,8,9,12,14,16,17,18,19}$. Eight disjoint representatives are obtained from row supports
\begin{equation}
{0},\ {1,5},\ {2,10},\ {11,13},\ {16,17}, {8,19},\ {6,9,12},\ {4,7,18},
\end{equation}
each of whose generator columns sum to $e_0$. To rule out a ninth, we exhaustively enumerate the $2^{15}$ vectors in
$e_0+\operatorname{im}H^T$.}

\begin{remark}[Encoding rate]
La-cross codes are not constant rate since at fixed $k$ they encode $K=k^2$ logical qubits with $N\approx2n^2$ physical qubits. In the primitive examples, $n=s(2^k-1)$, so $K=O((\log N)^2)$ and $D=\Theta(\sqrt N)$. Thus their appeal lies in their finite size performance rather than their asymptotic rate. At small sizes, La-cross codes provide compact memories with good performance and low connectivity requirements \cite{pecorari2025high}. In the remainder of this work we show that they also possess useful computational structure: fast logical measurements from disjoint representatives, a full logical Clifford group generated by seed automorphisms and fold transversal gates, and an addressable surface code interface for injecting magic states.
\end{remark}

\section{Automorphisms and the Clifford group}\label{sec:clifford}

In this section we use symmetries of the classical La-cross seed to obtain logical gates on the HGP code. An automorphism of the classical code lifts to a physical transformation of the product code and induces a linear action on its logical qubits \cite{berthusen2025automorphism}. We determine this logical action for general open La-cross seeds and show that, for primitive complete-period seeds, it gives the full group $\GL(k,2)$ independently on each of the the two logical tensor factors.

These automorphisms are particularly appealing fault-tolerant operations since on the principal sector (where logical operators are supported) they act only by qubit permutations, which may be implemented by relabelling or atom movement without spreading errors \cite{grassl2013leveraging,sayginel2025fault,berthusen2025automorphism}. The corresponding action on the check sector may additionally require the application of a \textsf{CNOT} circuit. Together with the fold gates introduced below, we show that for the primitive complete-period family automorphisms generate the full logical Clifford group (\cref{thm:clifford}).

\subsection{Classification for an arbitrary open seed}

Consider an arbitrary open La-cross seed
$C_{k,n}=\ker H$, with $n>k\ge2$ and denote by $g_j=Ge_j\in\F_2^k$ the $j$th column of the seed generator matrix. Define its
multiplicity function by
\begin{equation}
    \mu_n(v)=\#\{j:g_j=v\}, \qquad v\in\F_2^k.
    \label{eq:column-multiplicity}
\end{equation}
A coordinate permutation $P\in\Sym_n$ is an automorphism of the seed code if it preserves $C_{k,n}$. Since the rows of $G$ form a basis of $C_{k,n}$, this is equivalent to the existence of a unique $A\in\GL(k,2)$ such that
\begin{equation}
    AG=GP.
\end{equation}
The matrix $A$ is the induced action of $P$ on the $k$-dimensional seed-code coordinates.  We denote the group of such induced transformation
by
\begin{equation}
    \Gamma_{k,n} = \left\{ A\in\GL(k,2): AG=GP \text{ for some }P\in\Sym_n  \right\}.
    \label{eq:gamma}
\end{equation}
Equivalently, since a coordinate permutation only reorders the columns
of $G$,
\begin{equation}
    \Gamma_{k,n}  = \left\{A\in\GL(k,2):\mu_n(Av)=\mu_n(v)\text{ for all}v\in\F_2^k \right\}.
    \label{eq:multiplicity}
\end{equation}
Thus the logical automorphism group is determined entirely by the multiset of columns of $G$. For La-cross codes in particular, the recurrence relation \eqref{eq:column-recurrence} makes this multiset explicit, allowing us to determine $\Gamma_{k,n}$.

\subsection{Periods and truncation}
Recall from \cref{eq:polmaps} that, under the identification $\F_2^k \simeq\F_2[x]/(h_k)$, the columns of $G$ are given by
\begin{equation}
    g_j=x^j\bmod h_k .
\end{equation}
Since the constant term of $h_k$ is nonzero, $x$ is invertible in this quotient. Let $\ell=\text{ord}(x)$ be its multiplicative order and write
\begin{equation}
    \mathcal O=\{g_0,\ldots,g_{\ell-1}\}
\end{equation}
for the resulting orbit of $g_0$ under repeated multiplication by $x$. For $n=s\ell+r$, with $0\le r<\ell$, the column multiset consists of $s$ complete copies of $\mathcal O$ followed by the prefix
\begin{equation}
    \mathcal P_r=\{g_0,\ldots,g_{r-1}\}.
\end{equation}
Since vector in $\mathcal{P}_r$ occur with multiplicity $s+1$, vectors in $\mathcal{O}\setminus \mathcal{P}_r$ with multiplicity $s$,  and vectors outside $\mathcal{O}$ with multiplicity zero, equation \eqref{eq:multiplicity} gives,
\begin{equation}
    \Gamma_{k,n} = \text{Stab}_{\GL(k,2)}(\mathcal O) \cap \text{Stab}_{\GL(k,2)}(\mathcal P_r),
    \label{eq:period-group}
\end{equation}
for $s\ge1$, with the empty-prefix stabilizer understood as the full group.

In particular, if $h_k$ is primitive then $\mathcal O=\F_2^k\setminus\{0\}$ and for a complete period seed we have
\begin{equation}
    \Gamma_{k,s(2^k-1)}=\GL(k,2).
    \label{eq:primitive-gamma}
\end{equation}
Thus it is possible to realise every invertible linear transformation of the seed logical space as a coordinate automorphism.

\begin{example}[Truncated $k=3$ seeds]
This example keeps the recurrence polynomial primitive, but relaxes the complete period assumption since the final period to be truncated. For $k=3$, the recurrence has period seven, and the seven columns in a complete period are the nonzero vectors of $\F_2^3$. Thus $\GL(3,2)$ acts on them as the automorphism group of the Fano plane. Writing $r=n\bmod 7$, the group $\Gamma_{3,n}$ is the stabiliser of the $r$ columns occurring one additional time.

For $r=0$ this gives the full group $\GL(3,2)$. For $r=1$ the stabiliser of a point is $S_4$, while for $r=2$ the stabiliser of an unordered pair is $D_8$. The first three recurrence vectors are linearly independent, so for $r=3$ the distinguished subset is a triangle in the Fano plane; its stabilizer is $S_3$. The cases
$r=4,5,6$ follow by taking complements. To summarise:
\begin{equation}
\begin{array}{c|rrrrrrr}
r&0&1&2&3&4&5&6\\ \hline
\Gamma_{3,n}
&\GL(3,2)&S_4&D_8&S_3&S_3&D_8&S_4\\
|\Gamma_{3,n}|
&168&24&8&6&6&8&24
\end{array}
\end{equation}
Hence, non-trivial logical automorphism groups also exist away from complete periods, although their size decreases sharply for $r\neq 0$. The groups are symmetric under $r\leftrightarrow 7-r$ since the complement of $P_{7-r}$ is a cyclic translation of $P_r$, so their setwise stabilizers are conjugate in $\operatorname{GL}(3,2))$ Consequently $\Gamma_{3,7s+r}\cong \Gamma_{3,7s+(7-r)}$ as we see exemplified in the above table.
\end{example}

\begin{example}[A reducible complete period seed]
In this example we instead keep the complete period assumption but relax primitivity of the recurrence polynomial. Consider
\begin{equation}
    h_5(x) =(x^2+x+1)(x^3+x^2+1),
\end{equation}
Under the Chinese remainder
decomposition
\begin{equation}
    \F_2[x]/(h_5)\cong U_2\oplus U_3,
\end{equation}
where \begin{equation}
U_2=\F_2[x]/(x^2+x+1)\cong\F_4,
\qquad
U_3=\F_2[x]/(x^3+x^2+1)\cong\F_8.
\end{equation}
The residue class of $x$ has order $3$ in $U_2$ and order $7$ in $U_3$, and hence has order $\text{lcm}(3,7)=21$ in the product. The recurence therefore has period $21$, and its complete orbit is
\begin{equation}
    \mathcal O = (U_2\setminus\{0\}) \times (U_3\setminus\{0\}),
\end{equation}
and hence for a complete-period seed, $\Gamma_{5,21s}=\operatorname{Stab}_{\GL(5,2)}(\mathcal O)$. Preserving $\mathcal O$ forces the two coordinate subspaces $U_2\oplus\{0\}$ and $\{0\}\oplus U_3$ to be preserved separately, and hence
\begin{equation}
    \Gamma_{5,21s} \cong \GL(2,2)\times\GL(3,2).
\end{equation}
This provides a simple example in which a substantial logical automorphism group survives even though $h_5$ is reducible.
\end{example}

\subsection{Lifting seed automorphisms}

We now lift the classical seed automorphisms classified above to logical gates on the La-cross code. Let $P$ be a coordinate automorphism of the classical seed. Since $P$ preserves $C_{k,n}=\ker H$, we have
\begin{equation}
\ker(HP)=\ker H,
\end{equation}
and hence $HP$ and $H$ have the same row space. Therefore there exists a unique $Q\in\GL(m,2)$ such that $HP=QH$.

When $Q$ is itself a permutation, $(P,Q)$ is a symmetry of the Tanner graph. In general $Q$ is an arbitrary invertible linear map, so the action on the check sector is implemented by a \textsf{CNOT} circuit rather than just a qubit relabelling. The Tanner graph symmetry group is in general then smaller than the logical automorphism group; examples are given in appendix \ref{app:aut-check}.

Applying two seed automorphisms $(P_1,Q_1)$ and $(P_2,Q_2)$ to the two factors of the HGP gives the physical binary $X$ map
\begin{equation}
    T = (P_1\otimes P_2) \oplus (Q_1^{-T}\otimes Q_2),
    \label{eq:automorphism}
\end{equation}
with the corresponding $Z$ map $T^{-T}$.  Using $HP_i=Q_iH$, one verifies directly that these maps preserve the $X$- and $Z$-check spaces of \cref{eq:checks}. If
\begin{equation}
    GP_i=A_iG,
\end{equation}
then its logical action is
\begin{equation}
    \Xl(M)\mapsto\Xl(A_1MA_2^{-1}), \qquad \Zl(N)\mapsto\Zl(A_1^{-T}NA_2^T).
    \label{eq:aut-logical}
\end{equation}
Thus the HGP inherits the logical action of $\Gamma_{k,n}\times\Gamma_{k,n}$.

For primitive complete period seeds, $\Gamma_{k,n}=\GL(k,2)$, so these gadgets supply all transformations
\begin{equation}
    M\longmapsto AMB^T, \qquad A,B\in\GL(k,2),
\end{equation}
on the logical matrix labels.

\subsection{Fold gates and Clifford generation}

The automorphisms above provide a large family of logical linear transformations, but not the full Clifford group. The remaining ingredient comes from the square HGP structure. Exchanging the two identical seed factors defines a natural transpose fold, which supports two fold transversal Clifford operations \cite{quintavalle2023partitioning}.

Let $\tau$ denote transposition on both physical square arrays and on logical matrix labels,
\begin{equation}
    \tau(M)=M^T.
\end{equation}
The Hadamard fold applies a Hadamard gate to every physical qubit and exchanges each off-diagonal mirror pair. It swaps $X$- and $Z$-checks and, writing a logical Pauli as
\begin{equation}
    \Xl(M)\Zl(N), \qquad M,N\in\F_2^{k\times k},
\end{equation}
induces the logical action
\begin{equation}
    F_\tau:(M,N)\longmapsto \bigl(\tau(N),\tau(M)\bigr) = (N^T,M^T).
    \label{eq:hadamard-fold}
\end{equation}
The phase fold applies $S$ on the principal diagonal, $S^\dagger$ on the check-sector diagonal, and $\mathsf{CZ}$ between each off-diagonal mirror pair. It maps
\begin{equation}
    S_X(a,j)\longmapsto S_X(a,j)S_Z(j,a),
\end{equation}
and induces the logical action
\begin{equation}
    P_\tau:(M,N)\longmapsto \bigl(M,N+\tau(M)\bigr) = (M,N+M^T).
    \label{eq:phase-fold}
\end{equation}
Thus, in the canonical logical basis, $P_\tau$ acts by $S$ on diagonal logical qubits and by $\textsf{CZ}$ between transposed off-diagonal pairs. For the symplectic description below, we identify the $k\times k$ logical matrix labels with vectors in $\F_2^{k^2}$ in the basis $\{E_{ij}\}$.

\begin{theorem}[Clifford generation]\label{thm:clifford}
For a primitive complete period La-cross code with $k\ge3$, the seed automorphism gadgets together with $F_\tau$ and $P_\tau$ generate the full logical Clifford group modulo Pauli corrections,
\begin{equation}
    \Sp(2k^2,2).
\end{equation}
\end{theorem}

\begin{proof}
Write $D(R)=\text{diag}(R,R^{-T})$. Since
$\Gamma_{k,n}=\GL(k,2)$, the automorphism gadgets provide
$D(A\otimes B)$ for arbitrary $A,B\in\GL(k,2)$. Conjugating
$P_\tau$ by these transformations gives lower symplectic shears whose symmetric matrices include
\begin{equation}
    B_C=(C\otimes C^T)\tau, \qquad C\in\GL(k,2).
    \label{eq:phase-orbit}
\end{equation}
These span the full space of symmetric $k^2\times k^2$ matrices, as shown in detail in appendix \ref{app:phase-orbit}. Products therefore generate every lower symplectic shear. Conjugation by $F_\tau$ gives every upper shear, and together these generate the full symplectic group.
\end{proof}

\begin{remark}[Fault distance.]
Assume that qubit permutations are implemented by relabelling or atom movement rather than by \textsf{SWAP} circuits. Automorphisms and the Hadamard fold act on the principal sector only by permutations and single qubit gates, so by \cref{lem:principal} a nontrivial logical fault requires at least $D$ faults, even though errors may spread within the check sector \textsf{CNOT} circuits. The phase fold can spread a single fault across a mirror pair and therefore gives the general bound $d_{\rm fault}\ge\lceil D/2\rceil$; we therefore must separate successive $S$ folds by error correction. 
\end{remark}

\begin{remark}
Recent no-go results constrain the logical Clifford groups implementable directly by transversal, fold-transversal, and code automorphism gadgets \cite{chakraborty2026no,guyot2026addressability,koh2026achieving}. There is no contradiction with the theorem above our seed automorphisms may require general check-sector \textsf{CNOT} circuits, and arbitrary Cliffords are compiled from multiple such gadgets and fold layers with error correction between successive phase folds.
\end{remark}

\section{Logical measurements by parallel gauging}
\label{sec:measure}

Logical Pauli measurements are a fundamental computational primitive, underlying schemes such as Pauli based computation. We implement them using the  gauging construction of Williamson and Yoder \cite{williamson2026low}. The ingredient specific to La-cross codes is then the disjoint representative property that we prove holds (\cref{prop:disjoint}) for the primitive complete period family. These disjoint representatives allow several equivalent realisations of the same logical Pauli to be gauged simultaneously, leading to a natural space--time tradeoff between ancillary qubits and the number of syndrome-extraction rounds. Recent work has also developed fast and parallel logical measurements for qLDPC codes using, for example, product and lifted surgery constructions \cite{gu2026qgpu,berent2026lifted,zhou2026parallel,bhardwaj2026high,baspin2025fast}.

\begin{figure}[!htbp]
    \centering
    \begin{tikzpicture}[lc base,x=1cm,y=1cm]
\path[use as bounding box] (0,0) rectangle (12.8,4.2);
\node[lc title] at (0,3.95) {(a)};
\node[lc title] at (7.3,3.95) {(b)};
\begin{scope}[shift={(.6,3.5)},x=.43cm,y=-.43cm]
 \foreach \i in {0,...,6}{\foreach \j in {0,...,6}{\node[lc dot] at (\j,\i) {};}}
 \foreach \i/\clr in {0/lcOrange,1/lcBlue,2/lcTeal,3/lcOrange,4/lcTeal,5/lcPurple,6/lcPurple}{\foreach \j in {0,1,2,5}{\fill[\clr] (\j,\i) circle[radius=3pt];}}
\end{scope}
\foreach \y/\clr/\idx in {3.25/lcBlue/1,2.55/lcOrange/2,1.85/lcTeal/3,1.15/lcPurple/4}{
 \fill[\clr] (4.2,\y) circle[radius=3pt];\node[anchor=west,lc note] at (4.5,\y) {$P_{\idx}$};}
\node[lc note] at (2.5,.35) {four equivalent representatives};
\coordinate (v) at (8,3.3);\coordinate (w) at (10.6,3.3);
\coordinate (t) at (8,1.1);\coordinate (z) at (10.6,1.1);
\draw[lcGray,line width=1pt] (v)--(w)--(z)--(t)--cycle;
\draw[lcBlue,line width=2pt] (9.3,3.3)--(v)--(8,2.2);
\node[circle,fill=lcBlue,minimum size=7pt,inner sep=0pt] at (v) {};
\foreach \p in {w,t,z}{\node[circle,fill=lcInk,minimum size=6pt,inner sep=0pt] at (\p) {};}
\node[lc aux,fill=lcBlue!20,draw=lcBlue] at (9.3,3.3) {};
\node[lc aux,fill=lcBlue!20,draw=lcBlue] at (8,2.2) {};
\node[lc aux] at (10.6,2.2) {};\node[lc aux] at (9.3,1.1) {};
\node[lc note,text=lcBlue,anchor=east] at (7.8,3.3) {$A_v$};
\node[lc note] at (9.3,.35) {$\prod_v A_v=P$};
\end{tikzpicture}
    \caption{\textbf{Parallel gauging of a rank-one logical.} (a) Four disjoint supports represent the same rank-one logical operator in the $[[65,9,4]]$ code. (b) Schematic gauging of one representative: edge ancillas are introduced on an auxiliary graph and the product of the vertex checks measures the supported logical operator.  Explicit auxiliary graphs satisfying the fault-tolerance conditions are given in
    \cref{ex:measurement-65}.}
    \label{fig:measurement}
\end{figure}
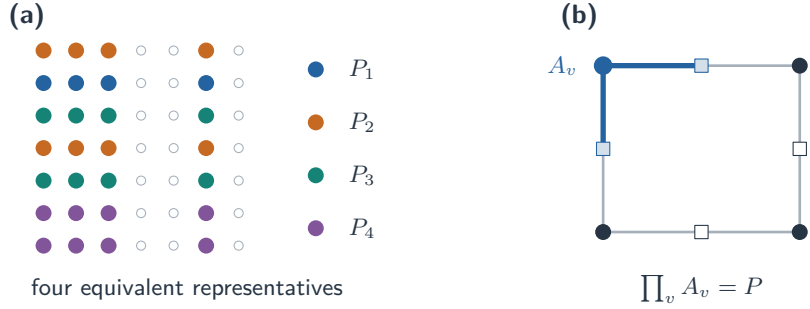

\subsection{Gauging a logical representative}
Fix nonzero $q,u\in\F_2^k$ and consider the operator
\begin{equation}
    \overline P=\Xl(qu^T).
\end{equation}
For any $v$ satisfying $Gv=q$, $P=X(v c_u^T,0)$ is a physical representative of $\overline P$. Let $S$ denote its support. To measure $P$, choose a connected auxiliary graph whose vertices include the qubits in $S$, place an ancilla in $\ket0$ on each edge, and measure the Gauss law checks
\begin{equation}
    A_v = X_v\prod_{e\ni v}X_e, \qquad \prod_v A_v=P,
    \label{eq:gauss}
\end{equation}
where $X_v$ denotes the data qubit Pauli at vertices corresponding to
$\supp(P)$, and $X_v=I$ at any vertex not associated with a
data qubit. The product of the measurement outcomes therefore gives the eigenvalue of $\overline P$. Original $Z$ checks are dressed by edge operators pairing their intersections with $S$, and cycle checks are added to complete the stabiliser group of the deformed code \cite{williamson2026low}.

Fault tolerance depends on the auxiliary graph. The construction of \cite{williamson2026low} requires, in particular, bounded degree and congestion, a bounded weight cycle basis, and sufficient expansion, which may be expressed as a Cheeger bound $h(G)\ge1$. Under these conditions the deformed code remains LDPC and retains the distance of the original code. Suitable graphs exist for a weight-$w$ logical using $O(w\log^3w)$ ancillas \cite{williamson2026low}, improved to $O(w\log w)$ in \cite{yuan2026parsimonious}. Repeated syndrome extraction for $O(D)$ rounds then gives the standard fault tolerant gauging measurement. 

\subsection{Parallel gauging from disjoint representatives}
For a primitive complete period La-cross code, \cref{prop:disjoint} provides $D$ pairwise disjoint vectors $v_a$ satisfying $Gv_a=q$.
Hence
\begin{equation}
    P_a=X(v_a c_u^T,0), \qquad 1\le a\le D,
\end{equation}
are disjoint physical representatives of the same logical operator $\overline P=\Xl(qu^T)$, each of weight either $D$ or $2D$. Their auxiliary systems can therefore be gauged in parallel without sharing data qubits.

This redundancy suggests a direct space--time tradeoff for logical measurement. For an integer $t$, gauge $r=2t-1$ equivalent representatives in parallel for $T=\left\lceil D/t \right\rceil$ syndrome extraction rounds and combine their logical outcomes by majority vote. Heuristically, an incorrect majority outcome requires faults affecting at least $t$ of the $r$ representatives, while each representative is protected in the time direction by $T$ rounds, suggesting a protection of order $tT\gtrsim D$.

Since $D=s2^{k-1}$ is even, \cref{eq:primitive}, taking
\begin{equation}
    t=D/2,\qquad r=D-1,\qquad T=2
    \label{eq:tradeoff}
\end{equation}
gives a constant in time (two rounds) resource estimate using at most $O(D^2\log D)$ ancillary qubits.

\begin{example}[A rank one measurement]
\label{ex:measurement-65}
Consider the $[[65,9,4]]$ code illustrated in figure \ref{fig:measurement} with $q=e_1$ and $u=(1,1,1)^T$. The seed word has support
\begin{equation}
    J=\supp(c_u)=\{0,1,2,5\}.
\end{equation}
The four disjoint subsets
\begin{equation}
    \{1\},\qquad
    \{0,3\},\qquad
    \{2,4\},\qquad
    \{5,6\}
\end{equation}
have indicator vectors $v_a$ satisfying $Gv_a=e_1$.  Hence the corresponding supports $\supp(v_a)\times J$, of weights $4,8,8,8$, are disjoint representatives of
\begin{equation}
    \overline P = \Xl\!\bigl(e_1(1,1,1)\bigr) = \Xl_{1,0}\Xl_{1,1}\Xl_{1,2}.
\end{equation}

The seed checks intersect $J$ in
\begin{equation}
    \{0,1\},\qquad
    \{1,2\},\qquad
    \{2,5\},\qquad
    \varnothing .
\end{equation}
For the weight four representative, choose the auxiliary graph
\begin{equation}
    C_4=(0,1,2,5,0).
\end{equation}
Each nontrivial check intersection is an edge of this cycle. Since $C_4$ has degree two, a weight-four cycle basis and $h(C_4)=1$, the corresponding deformed code satisfies the conditions of \cite{williamson2026low} and retains distance four.

For each weight-eight representative, take two copies of this cycle, one on each occupied principal row, and join corresponding vertices:
\begin{equation}
    C_4\square K_2\cong Q_3.
\end{equation}
The cube has degree three, a bounded congestion weight four cycle basis, and $h(Q_3)=1$, so these deformed codes are also distance four.
\end{example}

\section{An addressable surface code interface}\label{sec:interface}

The Clifford operations and Pauli measurements of \Cref{sec:clifford,sec:measure} are not sufficient for universal quantum computation. To tackle this, we treat a surface code patch as an external resource state and construct an addressable homomorphic \textsf{CNOT} between the surface code and a chosen La-cross logical qubit. This allows, for example, a magic state $\ket{T}$ encoded in a surface code to be teleported into the La-cross block, providing the resource to implement universal computation.

\subsection{Construction and logical action}

Let \(H_\src\in\F_2^{(D-1)\times D}\) be the parity check matrix of the length $D$ repetition code\footnote{Explicitly,
$$
H_\src=
\begin{pmatrix}
1&1&0&\cdots&0\\
0&1&1&\ddots&\vdots\\
\vdots&\ddots&\ddots&\ddots&0\\
0&\cdots&0&1&1
\end{pmatrix},
\qquad
\ker H_\src=\Span\{(1,\ldots,1)\}.
$$}.
 The source surface code patch is then $\text{HGP}(H_S,H_S)$ with parameters $[[D^2+(D-1)^2,1,D]]$. Its logical $\overline X_S$ may be represented on any principal row. We map each such row onto a physical representative of the same chosen La-cross logical operator \(\Xl(qu^T)\) (see \Cref{sec:codes}). More precisely, consider maps
\begin{equation}
f_1,g_1:\F_2^D\to\F_2^n,
\qquad
f_0,g_0:\F_2^{D-1}\to\F_2^m
\end{equation}
satisfying the chain map conditions
\begin{equation}
Hf_1=f_0H_\src,
\qquad
g_1H_\src^T=H^Tg_0 .
\label{eq:chain}
\end{equation}
Here $f_1$ and $g_1$ determine the desired coupling in the principal sector, while $f_0$ and $g_0$ complete it on the check sector so that the stabilisers are appropriately preserved. Define the matrix $\Phi$ and the corresponding \textsf{CNOT} circuit:
\begin{equation}
\Phi=(g_1\otimes f_1)\oplus(g_0\otimes f_0), \qquad U_\Phi=\prod_{\Phi_{\ell \sigma}=1}\textsf{CNOT}_{\sigma\to\ell}\,.
\label{eq:interface}
\end{equation}
With $H_{X,\src/L}$ and $H_{Z,\src/L}$ as the source/La-cross HGP check matrices, \cref{eq:chain} gives
\begin{align}
\Phi H_{X,\src}^T&=H_{X,L}^T(g_0\otimes f_1),
\label{eq:stab-x}\\
\Phi^TH_{Z,L}^T&=H_{Z,\src}^T(g_1^T\otimes f_0^T),
\label{eq:stab-z}
\end{align}
\Cref{eq:stab-x} shows that $U_\Phi$ maps the source X checks to a product of source and La-cross X checks specified by $g_0\otimes f_1$, and equivalently for the Z checks of the La-cross with \Cref{eq:stab-z}. Thus $U_\Phi$ preserves the joint codespace, as desired.

The same conditions determine which La-cross logical is addressed. Since $H_S\mathbf 1=0$, we have $Hf_1\mathbf 1=0,$ and hence $f_1\mathbf 1=c_u$ for some $u\in\F_2^k$. Additionally, writing $v_a=g_1e_a$, the second condition gives
$G(v_a+v_{a+1})=0$, 
so $Gv_a=q$ for some $q\in\F_2^k$ independent of $a$. Thus each surface code row is mapped to a representative of the same logical operator $\Xl(qu^T)$ (\cref{eq:equivalent-log-pauli}).

\paragraph{Logical action.}
$U_\Phi$ is a \textsf{CNOT} circuit, so its logical action is symplectic and fixed by the $X$ images. $X$ spreads only from source to target, so $\Xl(M)\mapsto\Xl(M)$, and by the previous paragraph $\overline X_\src\mapsto\overline X_\src\Xl(qu^T)$. Preserving commutation relations then fixes the $Z$ images:
\begin{align}
\overline{X}_\src
&\longmapsto
\overline{X}_\src\,\Xl(qu^T),
&
\overline{Z}_\src
&\longmapsto
\overline{Z}_\src,
\label{eq:logical-action-S}\\
\Xl(M)
&\longmapsto
\Xl(M),
&
\Zl(N)
&\longmapsto
\overline{Z}_\src^{\,q^TNu}\,\Zl(N),
\label{eq:logical-action}
\end{align}
which is a logical controlled-$\Xl(qu^T)$ with the surface code as control. For $q=e_i$ and $u=e_j$, this is a single \textsf{CNOT} onto logical qubit $(i,j)$. In general,
\begin{equation}
\Xl(qu^T)=\prod_{q_i u_j=1}\Xl_{ij},
\end{equation}
so it is a product of commuting \textsf{CNOT}s in the canonical basis, or a single \textsf{CNOT} in an adapted logical basis.

\subsection{Explicit construction}

We now restrict to the primitive complete period family and realise prescribed nonzero $q,u$. The source distance is then fixed to $D=d(C_{k,n})=s2^{k-1}$ (\cref{eq:primitive}): one source column per element of $\supp(c_u)$ and one source row per disjoint $v_a$. Every nonzero $c_u$ has weight $D$, and \cref{prop:disjoint} and its proof (\cref{app:disjoint-proof}) provide $D$ pairwise disjoint vectors $v_0,\ldots,v_{D-1}$ satisfying $Gv_a=q$, with $\wt(v_a)\le2$ and supports partitioning the seed coordinates, each support lying inside a single period. We order them period by period. Writing
\begin{equation}
\supp(c_u)=\{p_0<\cdots<p_{D-1}\},
\qquad
c_u^{\le b}=\sum_{b'\le b}e_{p_{b'}},
\end{equation}
we set
\begin{align}
f_1e_b&=e_{p_b},
&
f_0e_b&=Hc_u^{\le b},
\notag\\
g_1e_a&=v_a,
&
H^Tg_0e_a&=v_a+v_{a+1}.
\label{eq:explicit-maps}
\end{align}
These satisfy \cref{eq:chain}: $g_0$ is well defined because $v_a+v_{a+1}\in\ker G=\im H^T$ and $H^T$ is injective, and $f_0H_\src=Hf_1$ holds by construction of $c_u^{\le b}$ and since $Hc_u=0$. Geometrically, in the principal sector, source qubit $(a,b)$ controls \textsf{CNOT}s onto $\supp(v_a)\times\{p_b\}$, so a complete source row is copied onto $v_ac_u^T$, a representative of $\Xl(qu^T)$. Neighbouring rows are therefore mapped to equivalent representatives; the check-sector maps implement exactly the stabilisers relating them. We now state two properties of this construction with \Cref{prop:depth} and \Cref{lem:propagation}.

\begin{proposition}[Bounded-depth interface]\label{prop:depth}
The interface of \cref{eq:explicit-maps} has \textsf{CNOT} depth $O_k(1)$ and uses $O_k(D^2)$ \textsf{CNOT}s, where $O_k$ hides constants depending only on $k$, not on $s$, $q$ or $u$.
\end{proposition}

\begin{proof}
See Appendix \ref{app:depth_proof}.
\end{proof}

\begin{lemma}[Row--column propagation]\label{lem:propagation}
Let $a(i)$ be the unique $a$ with $i\in\supp(v_a)$. Under any subset of the \textsf{CNOT}s of $U_\Phi$, an $X$ in source column $b$ spreads only into target column $p_b$, a $Z$ in target row $i$ spreads only into source row $a(i)$, and the check sectors stay decoupled from the principal sectors.
\end{lemma}
\begin{proof}
$a(i)$ is well defined since the $v_a$ partition the seed coordinates. All \textsf{CNOT}s point from source to target, so propagation is a single step along the block-diagonal $\Phi$. In the principal sector, $(g_1\otimes f_1)(e_a\otimes e_b)=v_a\otimes e_{p_b}$, so source qubit $(a,b)$ controls \textsf{CNOT}s onto exactly $\supp(v_a)\times\{p_b\}$. Conversely, $(g_1\otimes f_1)^T(e_i\otimes e_j)=e_{a(i)}\otimes f_1^Te_j$, so target qubit $(i,j)$ receives at most one \textsf{CNOT}, from $(a(i),b)$ with $p_b=j$.
\end{proof}

\subsection{Magic state transfer and fault distance}\label{sec:transfer}

The interface \textsf{CNOT} can be used to transfer any resource state from a surface code into the primitive La-cross block using a teleportation circuit (see \cref{fig:interface-transfer}). In particular, high-fidelity magic state $\ket T=(\ket0+e^{i\pi/4}\ket1)/\sqrt2$ prepared in surface codes, allows universal computation. These can be prepared, for example, by cultivation \cite{gidney2024magic, sahay2026fold} or distillation \cite{bravyi2005universal}.

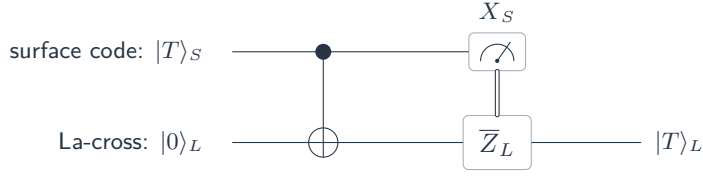
\begin{figure}[!htbp]
 \centering
 \begin{tikzpicture}[lc base,x=1cm,y=1cm]
\path[use as bounding box] (0,0) rectangle (11.0,2.6);
% labels
\node[anchor=east,lc note] at (2.65,1.85) {surface code: $\lvert T\rangle_S$};
\node[anchor=east,lc note] at (2.65,.65) {La-cross: $\lvert0\rangle_L$};
% measurement meter (X basis)
\node[lc box,minimum width=.75cm,minimum height=.5cm] (mx) at (6.4,1.85) {};
\draw[lcInk] ([xshift=-.22cm,yshift=-.12cm]mx.center) arc[start angle=180,end angle=0,radius=.22cm];
\draw[lcInk] ([yshift=-.12cm]mx.center) -- ++(55:.3cm);
\fill[lcInk] ([yshift=-.12cm]mx.center) circle[radius=.8pt];
\node[lc note,anchor=south] at (mx.north) {$\overline X_S$};
% classically controlled Pauli correction, directly below the meter
\node[lc box,minimum width=.9cm,minimum height=.5cm] (corr) at (6.4,.65) {$\overline Z_L$};
% quantum wires
\draw[lcInk] (2.9,1.85)--(mx.west);
\draw[lcInk] (2.9,.65)--(corr.west);
\draw[lcInk] (corr.east)--(8.3,.65);
% CNOT
\draw[lcInk] (4.1,.65)--(4.1,1.85);
\fill[lcInk] (4.1,1.85) circle[radius=3pt];
\draw[lcInk,fill=white] (4.1,.65) circle[radius=5.5pt];
\draw[lcInk] (3.91,.65)--(4.29,.65);\draw[lcInk] (4.1,.46)--(4.1,.84);
% classical wire
\draw[lcInk,double,double distance=1pt] (mx.south)--(corr.north);
% output
\node[lc note,anchor=west] at (8.35,.65) {$\lvert T\rangle_L$};
\end{tikzpicture}
 \caption{\textbf{Magic-state transfer.} The surface code patch controls an interface \textsf{CNOT} into an initialised La-cross logical qubit. Measuring $\Xl_\src$ and applying a $\Zl_L$ correction teleport $\lvert T\rangle$.}
 \label{fig:interface-transfer}
\end{figure}

\Cref{lem:propagation} shows that the interface never spreads an error beyond one principal row or column, while \cref{lem:principal} shows that every logical error occupies at least $D$ of them. To turn this into a fault distance, we consider the following protocol. Both blocks start in ideal encoded states and undergo $r\ge0$ rounds of bare-ancilla syndrome extraction, with any \textsf{CNOT} order. The interface $U_\Phi$ is then applied in any schedule, followed by $r$ further rounds on the target, an ideal round of target checks and a destructive $X$ readout of the source. A fault is an arbitrary Pauli on one qubit at one time step, an arbitrary two-qubit Pauli after one \textsf{CNOT}, or a flipped measurement outcome, including the source readout; decoding and frame updates are noiseless. Detectors compare consecutive check values and follow \cref{eq:stab-x,eq:stab-z} across the interface. The transfer fails if the frame bit, given by the readout of $\overline X_\src$, is wrong, or if the target is left with a nontrivial logical error. The fault distance is the minimum number of faults that trigger no detector and cause a failure.

\begin{proposition}[Transfer fault distance]\label{prop:transfer-fd}
The transfer through the interface of \cref{eq:explicit-maps} has fault distance $D$ in the model above, for any schedule of $U_\Phi$ and of the syndrome extraction circuits.
\end{proposition}

\begin{proof}
See appendix \ref{app:transfer-fd-proof}.
\end{proof}

\begin{example}[{$[[25,1,4]]\to[[65,9,4]]$ interface}]\label{ex:interface}
For $k=3$, $n=7$, $D=4$, $q=e_1$ and $u=(1,1,1)^T$, the maps
\begin{align}
  f_1&=(e_0,e_1,e_2,e_5), & f_0&=(e_0,e_1,e_2),\\
  g_1&=(e_0+e_3,\,e_1,\,e_2+e_4,\,e_5+e_6), & g_0&=(e_0,e_1,e_2+e_3)
\end{align}
satisfy \cref{eq:chain} and implement controlled-$\overline{X}(e_1(1,1,1))$, a single addressed qubit in an adapted basis, with $7\cdot4+4\cdot3=40$ \textsf{CNOT}s of maximum degree two, i.e.\ two layers of $25$ and $15$ gates. Here $(p_b)_b=(0,1,2,5)$ and $(a(i))_i=(0,1,2,0,2,3,3)$, and \cref{prop:transfer-fd} gives fault distance $4$ for any schedule (error propagation is illustrated in \cref{fig:interface_description}). An exhaustive Stim search over fault sets of weight up to $3$ finds no undetected failure for $r\in\{0,1,2,4\}$, for randomly sampled schedules of $U_\Phi$ and of the syndrome extraction.
\end{example}

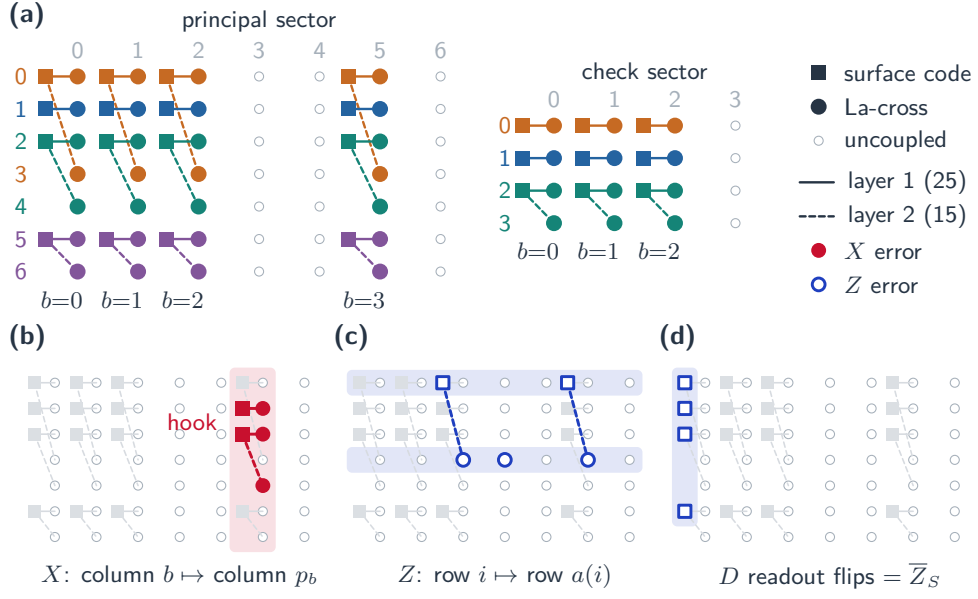
\begin{figure}
    \centering
    % Interface example in one figure: CNOT pattern (a) and fault distance (b)-(d).
% 40-CNOT interface [[25,1,4]] -> [[65,9,4]], q = e_1, u = (1,1,1).
%  A = (g_1 (x) f_1) (+) (g_0 (x) f_0). Source qubit (a,b) is drawn as a square
%  beside its target column (p_b, resp. b) at the height of its layer-1 target.
%  principal: v = ({0,3},{1},{2,4},{5,6}), p = (0,1,2,5); check: g_0 e_a = ({0},{1},{2,3}).
%  Target row labels are coloured by a(i) = (0,1,2,0,2,3,3) (principal) and by the
%  coupled source check row (0,1,2,2) (check sector).
% (b) source X hook (1,3),(2,3) -> target (1,5),(2,5),(4,5): stays in column 5.
% (c) target Z hook (3,2),(3,3),(3,5) -> source (0,2),(0,3): stays in row a(3)=0.
% (d) D = 4 readout flips on source column b=0 = Zbar_S, undetected.
% error colours (RGB convention: X red, Z blue), distinct from the row colours of panel (a)
\definecolor{lcXerr}{HTML}{C8102E}
\definecolor{lcZerr}{HTML}{1F3FBF}
\begin{tikzpicture}[lc base,x=1cm,y=1cm,
  e1/.style={line width=.9pt},
  e2/.style={line width=.9pt,dash pattern=on 2.2pt off 1.4pt},
  srcq/.style={rectangle,fill=#1,inner sep=0pt,minimum size=5.6pt},
  fe1/.style={lcGray!45,line width=.6pt},
  fe2/.style={lcGray!45,line width=.6pt,dash pattern=on 2pt off 1.3pt},
  fsrc/.style={rectangle,fill=lcGray!40,inner sep=0pt,minimum size=4.6pt},
  xsrc/.style={rectangle,fill=lcXerr,inner sep=0pt,minimum size=5.4pt},
  zsrc/.style={rectangle,draw=lcZerr,line width=1.1pt,fill=white,inner sep=0pt,minimum size=4.6pt},
  xe/.style={lcXerr,line width=1.1pt},
  ze/.style={lcZerr,line width=1.1pt},
  zed/.style={lcZerr,line width=1.1pt,dash pattern=on 2.2pt off 1.4pt},
  xed/.style={lcXerr,line width=1.1pt,dash pattern=on 2.2pt off 1.4pt},
  band/.style={rounded corners=2pt}]
\path[use as bounding box] (0,0) rectangle (12.8,8.1);
% ===================== (a) CNOT pattern =====================
\node[lc title] at (0,7.85) {(a)};
\begin{scope}[shift={(0,3.8)}]
\def\dx{.42}% source offset to the left of its target column
% ---------------- principal sector ----------------
\def\PX#1{1.05+.8*#1}\def\PY#1{3.25-.43*#1}
\foreach \j in {0,...,6}{
  \node[lc note,text=lcGray] at ({\PX{\j}},3.6) {\j};
  \foreach \i in {0,...,6}{\node[lc dot] at ({\PX{\j}},{\PY{\i}}) {};}}
% target row labels coloured by the source row that couples to them, a(i)
\foreach \i/\clr in {0/lcOrange,1/lcBlue,2/lcTeal,3/lcOrange,4/lcTeal,5/lcPurple,6/lcPurple}{
  \node[lc note,text=\clr] at (.3,{\PY{\i}}) {\i};}
% (source row a, layer-1 row, layer-2 row or -1, colour)
\foreach \j/\b in {0/0,1/1,2/2,5/3}{
  \foreach \ia/\ib/\clr in {0/3/lcOrange,1/-1/lcBlue,2/4/lcTeal,5/6/lcPurple}{
    \coordinate (s) at ({\PX{\j}-\dx},{\PY{\ia}});
    \draw[\clr,e1] (s)--({\PX{\j}},{\PY{\ia}});
    \ifnum\ib>-1
      \draw[\clr,e2] (s)--({\PX{\j}},{\PY{\ib}});
      \fill[\clr] ({\PX{\j}},{\PY{\ib}}) circle[radius=3pt];
    \fi
    \fill[\clr] ({\PX{\j}},{\PY{\ia}}) circle[radius=3pt];
    \node[srcq=\clr] at (s) {};}
  \node[lc note] at ({\PX{\j}-\dx/2},.3) {$b{=}\b$};}
\node[lc note] at ({\PX{3}},3.98) {principal sector};
% ---------------- check sector ----------------
\def\CX#1{7.35+.8*#1}\def\CY#1{2.6-.43*#1}
\foreach \j in {0,...,3}{
  \node[lc note,text=lcGray] at ({\CX{\j}},2.95) {\j};
  \foreach \i in {0,...,3}{\node[lc dot] at ({\CX{\j}},{\CY{\i}}) {};}}
% target check-row labels coloured by the source check row that couples to them
\foreach \i/\clr in {0/lcOrange,1/lcBlue,2/lcTeal,3/lcTeal}{
  \node[lc note,text=\clr] at (6.7,{\CY{\i}}) {\i};}
\foreach \j in {0,1,2}{
  \foreach \ia/\ib/\clr in {0/-1/lcOrange,1/-1/lcBlue,2/3/lcTeal}{
    \coordinate (s) at ({\CX{\j}-\dx},{\CY{\ia}});
    \draw[\clr,e1] (s)--({\CX{\j}},{\CY{\ia}});
    \ifnum\ib>-1
      \draw[\clr,e2] (s)--({\CX{\j}},{\CY{\ib}});
      \fill[\clr] ({\CX{\j}},{\CY{\ib}}) circle[radius=3pt];
    \fi
    \fill[\clr] ({\CX{\j}},{\CY{\ia}}) circle[radius=3pt];
    \node[srcq=\clr] at (s) {};}
  \node[lc note] at ({\CX{\j}-\dx/2},.95) {$b{=}\j$};}
\node[lc note] at ({\CX{1.5}},3.33) {check sector};
% ---------------- legend ----------------
\node[srcq=lcInk] at (10.85,3.3) {};
\node[anchor=west,lc note] at (11.05,3.3) {surface code};
\fill[lcInk] (10.85,2.85) circle[radius=3pt];
\node[anchor=west,lc note] at (11.05,2.85) {La-cross};
\node[lc dot] at (10.85,2.4) {};
\node[anchor=west,lc note] at (11.05,2.4) {uncoupled};
\draw[lcInk,e1] (10.6,1.85)--(11.05,1.85);
\node[anchor=west,lc note] at (11.1,1.85) {layer 1 (25)};
\draw[lcInk,e2] (10.6,1.4)--(11.05,1.4);
\node[anchor=west,lc note] at (11.1,1.4) {layer 2 (15)};
\fill[lcXerr] (10.85,.95) circle[radius=2.8pt];
\node[anchor=west,lc note] at (11.05,.95) {$X$ error};
\draw[lcZerr,line width=1.1pt,fill=white] (10.85,.5) circle[radius=2.4pt];
\node[anchor=west,lc note] at (11.05,.5) {$Z$ error};
\end{scope}
% ===================== (b)-(d) fault distance =====================
% faint copy of the principal interface (local coordinates)
\def\BX#1{.55*#1}\def\BY#1{3.0-.34*#1}\def\bdx{.27}
\def\faint{%
  \foreach \j in {0,...,6}{\foreach \i in {0,...,6}{\node[lc dot] at ({\BX{\j}},{\BY{\i}}) {};}}
  \foreach \j in {0,1,2,5}{\foreach \ia/\ib in {0/3,1/-1,2/4,5/6}{
    \draw[fe1] ({\BX{\j}-\bdx},{\BY{\ia}})--({\BX{\j}},{\BY{\ia}});
    \ifnum\ib>-1 \draw[fe2] ({\BX{\j}-\bdx},{\BY{\ia}})--({\BX{\j}},{\BY{\ib}});\fi
    \node[fsrc] at ({\BX{\j}-\bdx},{\BY{\ia}}) {};}}}
% ---- (b) X keeps its column ----
\node[lc title] at (0,3.6) {(b)};
\begin{scope}[shift={(.75,0)}]
 \fill[lcXerr!13,band] ({\BX{5}-\bdx-.17},{\BY{0}+.2}) rectangle ({\BX{5}+.17},{\BY{6}-.2});
 \faint
 \draw[xe] ({\BX{5}-\bdx},{\BY{1}})--({\BX{5}},{\BY{1}});
 \draw[xe] ({\BX{5}-\bdx},{\BY{2}})--({\BX{5}},{\BY{2}});
 \draw[xed] ({\BX{5}-\bdx},{\BY{2}})--({\BX{5}},{\BY{4}});
 \foreach \i in {1,2}{\node[xsrc] at ({\BX{5}-\bdx},{\BY{\i}}) {};}
 \foreach \i in {1,2,4}{\fill[lcXerr] ({\BX{5}},{\BY{\i}}) circle[radius=2.8pt];}
 \node[lc note,text=lcXerr,anchor=east] at ({\BX{5}-\bdx-.2},{\BY{1.5}}) {hook};
\end{scope}
\node[lc note] at (2.4,.45) {$X$: column $b\mapsto$ column $p_b$};
% ---- (c) Z keeps its row ----
\node[lc title] at (4.3,3.6) {(c)};
\begin{scope}[shift={(5.05,0)}]
 \fill[lcZerr!13,band] ({-\bdx-.17},{\BY{3}+.17}) rectangle ({\BX{6}+.17},{\BY{3}-.17});
 \fill[lcZerr!13,band] ({-\bdx-.17},{\BY{0}+.17}) rectangle ({\BX{6}+.17},{\BY{0}-.17});
 \faint
 \draw[zed] ({\BX{2}-\bdx},{\BY{0}})--({\BX{2}},{\BY{3}});
 \draw[zed] ({\BX{5}-\bdx},{\BY{0}})--({\BX{5}},{\BY{3}});
 \foreach \j in {2,3,5}{\draw[lcZerr,line width=1.1pt,fill=white] ({\BX{\j}},{\BY{3}}) circle[radius=2.4pt];}
 \foreach \j in {2,5}{\node[zsrc] at ({\BX{\j}-\bdx},{\BY{0}}) {};}
\end{scope}
\node[lc note] at (6.7,.45) {$Z$: row $i\mapsto$ row $a(i)$};
% ---- (d) tight witness ----
\node[lc title] at (8.6,3.6) {(d)};
\begin{scope}[shift={(9.35,0)}]
 \fill[lcZerr!13,band] ({-\bdx-.17},{\BY{0}+.2}) rectangle ({-\bdx+.17},{\BY{5}-.2});
 \faint
 \foreach \i in {0,1,2,5}{\node[zsrc] at ({-\bdx},{\BY{\i}}) {};}
\end{scope}
\node[lc note] at (11.0,.45) {$D$ readout flips $=\overline Z_S$};
\end{tikzpicture}
\caption{\textbf{Surface-code interface for
\( [[25,1,4]]\to[[65,9,4]] \)} with \(q=e_1\) and
\(u=(1,1,1)^T\).
(a) Two-layer CNOT pattern implementing the interface.
(b) \(X\)-error propagation preserves principal columns.
(c) \(Z\)-error propagation preserves principal rows.
(d) A weight-\(D\) source readout error realizes
\(\overline Z_\src\).}
    \label{fig:interface_description}
\end{figure}

\begin{remark}[Addressable \textsf{S} gates]
Let $F = F_\tau \otimes F_\src$ denote the Hadamard folds on both blocks.
The circuit $U_\Phi' = F U_\Phi F = \prod_{\Phi_{\ell \sigma}=1} \textsf{CNOT}_{\tau(\ell)\to\tau_\src(\sigma)}$ is a reversed, relabelled copy of $U_\Phi$. Conjugating \cref{eq:logical-action-S,eq:logical-action} by $F$ shows that it implements a logical CNOT with control $\Zl(uq^T)$ on the La-cross block, i.e.\ $\textsf{CNOT}_{(i,j)\to \src}$ for $(q,u) = (e_j, e_i)$. \Cref{lem:propagation} and \Cref{prop:transfer-fd} carry over with $X \leftrightarrow Z$ and rows $\leftrightarrow$ columns exchanged. Preparing the surface code in $\ket{Y}=\textsf{S}\ket{+}$, applying $U_\Phi'$ and measuring $\Zl_S$, followed by a $\Zl$ correction depending on the outcome, then implements an addressed logical \textsf{S}, providing an alternative to compiling an addressed \textsf{S} from the phase fold. This can additionally be used to teleport the \textsf{T} gate using $\ket{T}$ and $\ket{Y}$ states.
\end{remark}

\begin{remark}
A closely related surface code gadget is also developed in a forthcoming revision of \cite{pecorari2025addressable}, where it is used as an auxiliary gadget for addressable Clifford computation on La-cross logical qubits. Here the construction is aimed at transferring an externally prepared magic state into a chosen La-cross logical qubit, and we additionally give an explicit chain-map construction, a bounded-depth implementation and a proof of fault distance~$D$.
\end{remark}

\section{Discussion}

Our results extend the computational toolkit for La-cross codes by identifying a set of fast native primitives for Clifford operations, logical measurements, and resource-state injection. Seed automorphisms and fold-transversal gates generate the full logical Clifford group for the primitive complete period family, while disjoint logical representatives enable parallel gauging and a direct space--time tradeoff for logical measurements. The surface code interface then provides an addressable route for importing externally prepared resource states. Together, these ingredients give a complete Clifford+$T$ instruction set within a compact, low-connectivity architecture suited to neutral-atom hardware.

The present work is focused on small finite-size codes rather than asymptotically high-rate constructions. Several extensions are therefore natural. The primitive complete period condition is sufficient but not necessary for the automorphism and disjoint-representative properties, suggesting a broader class of useful La-cross instances. It will also be important to develop a full circuit-level analysis of the parallel measurement protocol and to quantify the resulting space--time tradeoffs under, e.g. realistic neutral atom noise and scheduling constraints.

\section*{Acknowledgements}

We thank the QPerfect research team for many stimulating discussions. We are grateful to Shannon Whitlock, Stanimir Kondov, and Jerome Dubail for numerous conversations related to this work. A.D.F. was supported by the European Commission via the Marie Skłodowska-Curie Action (MSCA) Gen-Q (Generation Quantum) programme, grant agreement number 101217386. L.P. was supported by the University of Strasbourg ITT program QMAT and the Interreg programme UpQuantVal. G.P was supported by the French National Research Agency under the Investments of the Future Program projects ANR-21-ESRE-0032 (aQCess), ANR-22-CE47-0013-02 (CLIMAQS), and ANR-22-CMAS-0001 France 2030 (QuanTEdu-France). T.H. was supported by the Cifre grant collaboration 24/RCH/374 between QPerfect and the University of Strasbourg. 

We also acknowledge the use of GPT-5.6 Sol and Claude Opus 4.8 for assistance with presentation, including figure generation, and for numerical checks of the fault distance of the surface code interface gadget. All technical claims and final results were independently checked and are the responsibility of the authors.

\appendix

\section{Disjoint logical representatives}
\label{app:disjoint-proof}

In this appendix we prove \cref{prop:disjoint}. We consider an $X$-type rank one logical $\Xl(qu^T)$ with nonzero $q,u\in\F_2^k$.

In each primitive period, the columns of $G$ run through the
nonzero vectors of $\F_2^k$ once. For fixed $q$, partition these vectors into the singleton $\{q\}$ and the two element orbits $\{x,x+q\}$ under translation by $q$. This gives
\begin{equation}
    1+\frac{2^k-2}{2}=2^{k-1}
\end{equation}
disjoint subsets. Let $v_a$ denote the coordinate indicator of one such subset. Then $Gv_a=q$, since the singleton contributes $q$, while each pair contributes $x+(x+q)=q$.

Repeating the construction independently in each of the $s$ periods gives $D=s2^{k-1}$ vectors $v_a$ with pairwise disjoint supports and $Gv_a=q$. Hence $(v_a c_u^T,0)$ are pairwise disjoint representatives of $\Xl(qu^T)$. Since $\wt(c_u)=D$ and each $v_a$ has weight one or two, each representative has weight $D$ or $2D$.

For any two choices,
\begin{equation}
    G(v_a+v_b)=0,
\end{equation}
so $v_a+v_b=H^Ty$ for some $y$, since
$\ker G=\operatorname{im}H^T$. Using $Hc_u=0$,
\begin{equation}
    \bigl((v_a+v_b)c_u^T,0\bigr) = \bigl(H^Tyc_u^T,yc_u^TH^T\bigr)
\end{equation}
is an $X$ stabiliser. Thus all representatives are equivalent. The $Z$ statement follows similarly.

\section{Automorphism and Clifford details}
\label{app:automorphisms}

\subsection{Finite automorphism data}
\label{app:aut-check}

To determine the classical seed automorphism group, we first compute the columns $g_j$ of the generator matrix $G$ and record how many times each vector $v\in\F_2^k$ occurs. A linear transformation $A\in\GL(k,2)$ is induced by a coordinate permutation precisely when it preserves these column multiplicities, that is,
\begin{equation}
\mu_n(Av)=\mu_n(v) \qquad \text{for all }v\in\F_2^k.
\end{equation}
So that rather than searching over all $n!$ permutations of the physical coordinates, we enumerate the much smaller group $\GL(k,2)$ and keep those elements satisfying this condition. This is the method we use to obtain the finite automorphism groups in the main body.

For $k=4$, the group order depends only on $r=n\bmod15$. We find
\begin{equation}
\begin{array}{c|rrrrrrrr}
r&0&1&2&3&4&5&6&7\\\hline
|\Gamma_{4,n}|&20160&1344&192&48&24&12&8&8\\[2pt]
r&8&9&10&11&12&13&14&\\\hline
|\Gamma_{4,n}|&8&8&12&24&48&192&1344&
\end{array}
\end{equation}

Some reducible examples are in addition:
\begin{center}
\begin{tabular}{rrrr}
\toprule
$(k,n)$ & $d(C_{k,n})$ & $|\Gamma_{k,n}|$
& $|\operatorname{Aut}_T(H)|$\\
\midrule
$(5,20)$&9&48&1\\
$(5,21)$&10&1008&1\\
$(5,22)$&10&48&1\\
$(5,42)$&20&1008&1\\
\bottomrule
\end{tabular}
\end{center}
These computations illustrate that the logical automorphism group can remain large even when the Tanner graph has no nontrivial permutation symmetry.

\subsection{Span of the phase-fold orbit}
\label{app:phase-orbit}

We prove here the spanning statement used in the proof of theorem \ref{thm:clifford}. Conjugating the phase fold by seed automorphisms produces lower symplectic shears with symmetric matrices $B_C$, indexed by $C\in\GL(k,2)$. To obtain arbitrary lower shears, it is enough to show that these matrices span the full space of symmetric operators on the $k^2$ logical qubits,
\begin{equation}
    \text{span}\{B_C:C\in\GL(k,2)\} = \Sym(\F_2^{k^2}).
\end{equation}

Let $K=k^2$. In the basis $E_{ij}$, the matrix associated with a conjugated phase fold has entries
\begin{equation}
    (B_C)_{(i,j),(p,q)}=C_{iq}C_{pj}.
\end{equation}
We relabel the indices by
\begin{equation}
    \mathcal R(B)_{(i,q),(p,j)} = B_{(i,j),(p,q)}.
\end{equation}
Then
\begin{equation}
    \mathcal R(B_C) = \text{vec}(C)\text{vec}(C)^T.
\end{equation}
Thus the spanning problem is equivalent to the rank one symmetric matrices $\text{vec}(C)\text{vec}(C)^T$, with $C\in\GL(k,2)$ spanning the full symmetric matrix space.

Equivalently, a linear functional annihilating all of these matrices would give a multilinear polynomial of degree at most two in the entries of $C$ that vanishes for every $C\in\GL(k,2)$. For $k=3$, we check this directly. Evaluating the degree at most two multilinear monomials on all $168$ elements of $\GL(3,2)$ gives full rank. Hence no nonzero such functional exists, and the matrices $B_C$ span the full symmetric space.

For $k>3$, the same conclusion follows by reduction to the $k=3$ case. Any chosen pair of matrix entries can be placed inside a $3\times3$ row--column block, while the remaining entries are fixed so that the complementary block is invertible. Restricting a putative annihilating polynomial to this family of matrices reduces the problem to the $3\times3$ case, forcing its linear and quadratic coefficients to vanish. Therefore
\begin{equation}
    \text{span}\{B_C:C\in\GL(k,2)\} = \Sym(\F_2^{k^2}),
\end{equation}
as required in the proof of \cref{thm:clifford}.

\section{Surface code interface}
\subsection{Proof of \Cref{prop:depth}}\label{app:depth_proof}
All CNOTs of $U_\Phi$ act from source to target and thus commute. By K\"{o}nig's edge-colouring theorem they fit in $\Delta(\Phi)$ layers of disjoint CNOTs, where $\Delta(\Phi)$ is the largest row or column weight of $\Phi$, i.e.\ the largest number of CNOTs on one qubit. Since weights multiply under $\otimes$, it suffices that every row and column of $f_1,g_1,f_0,g_0$ has $O_k(1)$ nonzero entries. This follows from locality: $f_1e_b=e_{p_b}$ and $g_1e_a=v_a$ are local by construction; $(f_0)_{jb}\neq0$ only if seed check $j$ straddles the cut after $p_b$, i.e.\ $j\le p_b<j+k$; and, writing $H^Ty\leftrightarrow h_k(x)y(x)$ as polynomials, $g_0e_a=(v_a+v_{a+1})/h_k$ is supported between the extreme positions of $v_a$ and $v_{a+1}$, hence within two adjacent periods since the $v_a$ are ordered period by period. Every image therefore lies within two periods of $p_b$ or $v_a$, and each period contains exactly $2^{k-1}$ of the $p_b$ and $2^{k-1}$ of the $v_a$, so all weights are bounded by a function of $k$ alone. Hence $\Delta(\Phi)=O_k(1)$, and the number of CNOTs is at most $\Delta(\Phi)$ times $D^2+(D-1)^2$, i.e.\ $O_k(D^2)$.

\subsection{Proof of \cref{prop:transfer-fd}}\label{app:transfer-fd-proof}
Treat a flipped source readout as a $Z$ before the measurement and propagate all faults to the end, leaving $E_S$ on the source and $E_L$ on the target. Detectors are chained across the interface by \cref{eq:stab-x,eq:stab-z}, and the final source $X$ checks are reconstructed from the readout. Summing the detectors along each stabiliser's time line then cancels all intermediate outcomes. Hence, for an undetected fault set, $E_L$ commutes with all target checks and the $Z$ part of $E_S$ commutes with all source $X$ checks. Failure therefore requires the $Z$ part of $E_S$ to be $\overline Z_S$, or $E_L$ to be a nontrivial logical. Since $d(C_{k,n})=D$, \cref{lem:principal} shows that this occupies at least $D$ source principal rows, or at least $D$ target principal columns ($X$ part) or rows ($Z$ part).

It thus suffices that each fault adds at most one to each of these counts. Bare-ancilla extraction does not propagate errors between data qubits. A single-qubit fault, a syndrome flip or a readout flip therefore trivially satisfies this. A hook lies in the support of one check, whose principal part is one column for an $X$ check and one row for a $Z$ check (\cref{eq:checks}). For a principal interface CNOT fault on $(\sigma,l)=((a,b),(i,p_b))$, $\Phi_{\ell \sigma}=1$ forces $a(i)=a$, so its $X$ parts share column $p_b$ and its $Z$ parts lie in target row $i$ and source row $a(i)$. Check-sector interface faults do not reach the principal sectors. \Cref{lem:propagation} preserves all of this across the remaining interface CNOTs. Hence fewer than $D$ faults cannot cause an undetected failure.

The bound is tight: flipping the readout of the $D$ source qubits in one principal column gives $\overline Z_S$. It commutes with all source $X$ checks, so it triggers no detector, but it flips the frame bit (\cref{fig:interface_description}(d)).

\bibliographystyle{unsrtnat}
\bibliography{refs}

@article{hong2024long,
  title={Long-range-enhanced surface codes},
  author={Hong, Yifan and Marinelli, Matteo and Kaufman, Adam M and Lucas, Andrew},
  journal={Physical Review A},
  volume={110},
  number={2},
  pages={022607},
  year={2024},
  publisher={APS}
}

@article{pecorari2025high,
  title={High-rate quantum LDPC codes for long-range-connected neutral atom registers},
  author={Pecorari, Laura and Jandura, Sven and Brennen, Gavin K and Pupillo, Guido},
  journal={Nature Communications},
  volume={16},
  number={1},
  pages={1111},
  year={2025},
  publisher={Nature Publishing Group UK London}
}

@inproceedings{panteleev2022asymptotically,
  title={Asymptotically good quantum and locally testable classical LDPC codes},
  author={Panteleev, Pavel and Kalachev, Gleb},
  booktitle={Proceedings of the 54th annual ACM SIGACT symposium on theory of computing},
  pages={375--388},
  year={2022}
}

@article{berthusen2025automorphism,
  title={Automorphism gadgets in homological product codes},
  author={Berthusen, Noah and Gullans, Michael J and Hong, Yifan and Mudassar, Maryam and Tan, Shi Jie Samuel},
  journal={arXiv preprint arXiv:2508.04794},
  year={2025}
}

@article{cohen2022low,
  title={Low-overhead fault-tolerant quantum computing using long-range connectivity},
  author={Cohen, Lawrence Z and Kim, Isaac H and Bartlett, Stephen D and Brown, Benjamin J},
  journal={Science Advances},
  volume={8},
  number={20},
  pages={eabn1717},
  year={2022},
  publisher={American Association for the Advancement of Science}
}

@article{berent2026lifted,
  title={Lifted surgery: Fast processing with QLDPC codes},
  author={Berent, Lucas and Cohen, Lawrence Z and Quintavalle, Armanda O},
  journal={arXiv preprint arXiv:2609.11723},
  year={2026}
}

@article{gu2026qgpu,
  title={QGPU: Parallel logic in quantum LDPC codes},
  author={Gu, Boren and Liu, Andy Zeyi and Quintavalle, Armanda O and Xu, Qian and Eisert, Jens and Roffe, Joschka},
  journal={arXiv preprint arXiv:2603.05398},
  year={2026}
}

@article{malcolm2026computing,
  title={Computing efficiently in QLDPC codes},
  author={Malcolm, Alexander J and Glaudell, Andrew N and Fuentes, Patricio and Chandra, Daryus and Schotte, Alexis and DeLisle, Colby and Haenel, Rafael and Ebrahimi, Amir and Roffe, Joschka and Quintavalle, Armanda O and others},
  journal={Nature Communications},
  volume={17},
  pages={7286},
  year={2026},
  publisher={Nature Publishing Group UK London}
}

@article{zhou2026parallel,
  title={Parallel algebraic surgery for qLDPC codes with constant shuttling depth},
  author={Zhou, Tian-Gang and Gu, Boren and Eisert, Jens and Zhao, Chen},
  journal={arXiv preprint arXiv:2609.39874},
  year={2026}
}

@article{bhardwaj2026high,
  title={High-rate qLDPC processors},
  author={Bhardwaj, Aditya and Ma, Muzhou and Meister, Nadine and King, Robbie and Bluvstein, Dolev and Preskill, John and Cain, Madelyn and Xu, Qian and Huang, Hsin-Yuan},
  journal={arXiv preprint arXiv:2607.28795},
  year={2026}
}

@article{baspin2025fast,
  title={Fast surgery for quantum LDPC codes},
  author={Baspin, Nou{\'e}dyn and Berent, Lucas and Cohen, Lawrence Z},
  journal={arXiv preprint arXiv:2510.04521},
  year={2025}
}

@inproceedings{grassl2013leveraging,
  title={Leveraging automorphisms of quantum codes for fault-tolerant quantum computation},
  author={Grassl, Markus and Roetteler, Martin},
  booktitle={2013 IEEE International Symposium on Information Theory},
  pages={534--538},
  year={2013},
  organization={IEEE}
}

@article{sayginel2025fault,
  title={Fault-tolerant logical Clifford gates from code automorphisms},
  author={Sayginel, Hasan and Koutsioumpas, Stergios and Webster, Mark and Rajput, Abhishek and Browne, Dan E},
  journal={PRX Quantum},
  volume={6},
  number={3},
  pages={030343},
  year={2025},
  publisher={APS}
}

@article{pecorari2025addressable,
  title={Addressable gate-based logical computation with quantum LDPC codes},
  author={Pecorari, Laura and Guerci, Francesco Paolo and Perrin, Hugo and Pupillo, Guido},
  journal={arXiv preprint arXiv:2511.06124},
  year={2025}
}

@article{tillich2013quantum,
  title={Quantum LDPC codes with positive rate and minimum distance proportional to the square root of the blocklength},
  author={Tillich, Jean-Pierre and Z{\'e}mor, Gilles},
  journal={IEEE Transactions on Information Theory},
  volume={60},
  number={2},
  pages={1193--1202},
  year={2013},
  publisher={IEEE}
}

@article{quintavalle2022reshape,
  title={Reshape: A decoder for hypergraph product codes},
  author={Quintavalle, Armanda O and Campbell, Earl T},
  journal={IEEE Transactions on Information Theory},
  volume={68},
  number={10},
  pages={6569--6584},
  year={2022},
  publisher={IEEE}
}

@article{quintavalle2023partitioning,
  title={Partitioning qubits in hypergraph product codes to implement logical gates},
  author={Quintavalle, Armanda O and Webster, Paul and Vasmer, Michael},
  journal={Quantum},
  volume={7},
  pages={1153},
  year={2023},
  publisher={Verein zur F{\"o}rderung des Open Access Publizierens in den Quantenwissenschaften}
}

@article{williamson2026low,
  title={Low-overhead fault-tolerant quantum computation by gauging logical operators},
  author={Williamson, Dominic J and Yoder, Theodore J},
  journal={Nature Physics},
  volume={22}, 
  pages={598–603},
  year={2026},
  publisher={Nature Publishing Group UK London}
}

@article{yuan2026parsimonious,
  title={Parsimonious Quantum Low-Density Parity-Check Code Surgery},
  author={Yuan, Andrew C and Cowtan, Alexander and He, Zhiyang and Lin, Ting-Chun and Williamson, Dominic J},
  journal={arXiv preprint arXiv:2603.05082},
  year={2026}
}

@article{gidney2024magic,
  title={Magic state cultivation: growing T states as cheap as CNOT gates},
  author={Gidney, Craig and Shutty, Noah and Jones, Cody},
  journal={arXiv preprint arXiv:2409.17595},
  year={2024}
}

@article{sahay2026fold,
  title={Fold-transversal surface code cultivation},
  author={Sahay, Kaavya and Tsai, Pei-Kai and Chang, Kathleen and Su, Qile and Smith, Thomas B and Singh, Shraddha and Puri, Shruti},
  journal={PRX Quantum},
  volume={7},
  number={3},
  pages={033006},
  year={2026},
  publisher={APS}
}

@article{bravyi2005universal,
  title={Universal quantum computation with ideal Clifford gates and noisy ancillas},
  author={Bravyi, Sergey and Kitaev, Alexei},
  journal={Physical Review A—Atomic, Molecular, and Optical Physics},
  volume={71},
  number={2},
  pages={022316},
  year={2005},
  publisher={APS}
}

@article{xu2024constant,
  title={Constant-overhead fault-tolerant quantum computation with reconfigurable atom arrays},
  author={Xu, Qian and Bonilla Ataides, J Pablo and Pattison, Christopher A and Raveendran, Nithin and Bluvstein, Dolev and Wurtz, Jonathan and Vasi{\'c}, Bane and Lukin, Mikhail D and Jiang, Liang and Zhou, Hengyun},
  journal={Nature Physics},
  volume={20},
  number={7},
  pages={1084--1090},
  year={2024},
  publisher={Nature Publishing Group UK London}
}

@article{xu2025fast,
  title={Fast and parallelizable logical computation with homological product codes},
  author={Xu, Qian and Zhou, Hengyun and Zheng, Guo and Bluvstein, Dolev and Ataides, J Pablo Bonilla and Lukin, Mikhail D and Jiang, Liang},
  journal={Physical Review X},
  volume={15},
  number={2},
  pages={021065},
  year={2025},
  publisher={APS}
}

@article{chakraborty2026no,
  title={No-go theorem on fault tolerant gadgets for multiple logical qubits},
  author={Chakraborty, Aranya and Gottesman, Daniel},
  journal={PRX Quantum},
  volume={7},
  number={3},
  pages={033027},
  year={2026},
  publisher={APS}
}

@article{guyot2026addressability,
  title={On the addressability problem on CSS codes},
  author={Guyot, J{\'e}r{\^o}me and Jaques, Samuel},
  journal={Quantum},
  volume={10},
  pages={2145},
  year={2026},
  publisher={Verein zur F{\"o}rderung des Open Access Publizierens in den Quantenwissenschaften}
}

@article{koh2026achieving,
  title={Achieving the limits of automorphism gates},
  author={Koh, Jin Ming and Majidy, Shayan and Chakraborty, Aranya and Gong, Anqi and Tan, Shi Jie Samuel and Yao, Norman Y},
  journal={arXiv preprint arXiv:2609.19250},
  year={2026}
}

@article{breuckmann2021quantum,
  title={Quantum low-density parity-check codes},
  author={Breuckmann, Nikolas P and Eberhardt, Jens Niklas},
  journal={PRX quantum},
  volume={2},
  number={4},
  pages={040101},
  year={2021},
  publisher={APS}
}

@article{bravyi2024high,
  title={High-threshold and low-overhead fault-tolerant quantum memory},
  author={Bravyi, Sergey and Cross, Andrew W and Gambetta, Jay M and Maslov, Dmitri and Rall, Patrick and Yoder, Theodore J},
  journal={Nature},
  volume={627},
  number={8005},
  pages={778--782},
  year={2024},
  publisher={Nature Publishing Group UK London}
}

@article{horsman2012surface,
  title={Surface code quantum computing by lattice surgery},
  author={Horsman, Dominic and Fowler, Austin G and Devitt, Simon and Meter, Rodney Van},
  journal={New Journal of Physics},
  volume={14},
  number={12},
  pages={123011},
  year={2012},
  publisher={IOP Publishing}
}

@article{cowtan2026parallel,
  title={Parallel logical measurements via quantum code surgery},
  author={Cowtan, Alexander and He, Zhiyang and Williamson, Dominic J and Yoder, Theodore J},
  journal={PRX Quantum},
  volume={7},
  number={2},
  pages={020325},
  year={2026},
  publisher={APS}
}

@article{di2026error,
  title={Error-detected surgery on Iceberg codes},
  author={Di Fini, Andrea and Crew, Samuel and Pecorari, Laura and Pupillo, Guido},
  journal={arXiv preprint arXiv:2608.06187},
  year={2026}
}
\end{document}